\documentclass[11pt]{article}%
\usepackage{amssymb}
\usepackage{array,multirow}
\usepackage{amsmath}
\usepackage[doublespacing]{setspace}
\usepackage{geometry}
\usepackage{amsfonts}
\usepackage{sectsty}
\usepackage{chicago}
\usepackage{graphicx}
\usepackage{boxedminipage}%
\providecommand{\U}[1]{\protect\rule{.1in}{.1in}}
\newtheorem{theorem}{Theorem}

\newtheorem{proposition}[theorem]{Proposition}

\newenvironment{proof}[1][Proof]{\noindent\textbf{#1.} }{\ \rule{0.5em}{0.5em}}
\begin{document}
%
\sectionfont{\bfseries\large\sffamily}%
%

\subsectionfont{\bfseries\sffamily\normalsize}%
%

\noindent
{\sffamily\bfseries\Large Causal inference with two versions of treatment}%
%

\noindent
 \textsf{Raiden B. Hasegawa, Sameer K. Deshpande, Dylan S. Small and Paul R.
 Rosenbaum}\footnote{Raiden Hasegawa and Sameer K. Deshpande are PhD students
 and Dylan Small and Paul Rosenbaum are professors in the Department of
 Statistics, Wharton School, University of Pennsylvania, Philadelphia, PA
 19104-6340 US.\ 31 May 2018. \textsf{raiden@wharton.upenn.edu}.}%


\noindent
\textsf{Abstract. \ Causal effects are commonly defined as comparisons of the
potential outcomes under treatment and control, but this definition is
threatened by the possibility that either the treatment or the control
condition is not well-defined, existing instead in more than one version.
\ This is often a real possibility in nonexperimental or observational studies
of treatments, because these treatments occur in the natural or social world
without the laboratory control needed to ensure identically the same treatment
or control condition occurs in every instance. \ We consider the simplest
case: either the treatment condition or the control condition exists in two
versions that are easily recognized in the data but are of uncertain, perhaps
doubtful, relevance, e.g., branded Advil vs. generic ibuprofen. \ Common practice does not address versions of treatment:
typically the issue is either ignored or explicitly stated but assumed to be
absent. \ Common practice is reluctant to address two versions of treatment
because the obvious solution entails dividing the data into two parts with two
analyses, thereby (i) reducing power to detect versions of treatment in each part, (ii) creating problems of
multiple inference in coordinating the two analyses, (iii) failing to report a
single primary analysis that uses everyone. \ We propose and illustrate a new
method of analysis that begins with a single primary analysis of everyone that
would be correct if the two versions do not differ, adds a second analysis
that would be correct were there two different effects for the two versions,
controls the family-wise error rate in all assertions made by the several
analyses, and yet pays no price in power to detect a constant treatment effect in the primary analysis of everyone. \
Our method can be applied to analyses of constant additive treatment effects on
continuous outcomes. \ Unlike conventional simultaneous inferences, the new method is coordinating
several analyses that are valid under different assumptions, so that one
analysis would never be performed if one knew for certain that the assumptions
of the other analysis are true. \ It is a multiple assumptions problem, rather
than a multiple hypotheses problem. We discuss the relative merits of the 
method with respect to more conventional approaches to analyzing multiple comparisons.
\ The method is motivated and illustrated using a study of the possibility 
that repeated head trauma in high school football causes an increase in risk 
of early on-set cognitive decline.}%


\section{What are versions of treatment?}

Commonly, the effect on an individual caused by a treatment is defined as a
comparison of the two potential outcomes that this individual would exhibit
under treatment and under control; see Neyman (1923), Welch (1937) and Rubin
(1974). \ Implicit in this definition is the notion that the treatment and
control conditions are each well-defined. \ In particular, it is common to
assume that there are \textquotedblleft no versions of treatment or
control\textquotedblright; see Rubin (1986).\

By definition, versions of treatment are not intended additions to a study
design, but rather potential flaws in the study design. \ Versions of
treatment or control are often associated with finding treatments that occur
naturally, rather than experimentally manipulating a tightly controlled,
uniform treatment. \ Versions of one treatment should be distinguished from
the intentional study of distinct, competing treatments. \ When an
investigator discusses versions of one treatment, she is expressing a
preference for the conception that there is a single treatment, but is
acknowledging the possibility that her preferred conception is mistaken.
\ Branded Advil and generic ibuprofen are versions of one treatment ---
possibly different, but very plausibly expected to be the same --- whereas
ibuprofen and aspirin are different competing treatments. \ The investigator
and her audience prefer a primary analysis that does not distinguish versions
of treatment, but both would be reassured by evidence that showed their
preferred analysis does not embody a consequential error. \ Versions of
control groups should also be distinguished from the deliberate use of two
carefully selected control groups intended to reveal unmeasured biases if
present; see, for instance, Rosenbaum (1987). \ In particular, Campbell (1969)
suggested that two control groups should be deliberately selected to
systematically vary a specific unmeasured covariate in an effort to
demonstrate its irrelevance; however, versions of control are unintended flaws
in study design, not purposeful quasi-experimental devices.

There are two versions of either the treatment condition or the control
condition if we recognize in available data either two types of treated
subjects or two types of controls, but we are uncertain about, or perhaps
explicitly doubt, the relevance of this visible distinction. \ Versions refer
to a visible but perhaps unimportant distinction, not to a distinction that is
hidden or latent. \ There are important methodological issues in recognizing
treatments that inexplicably affect some people but not others; however, this
is practically and mathematically a different problem (Conover and Salsberg
1988; Rosenbaum 2007a).

In discussing randomized clinical trials, Peto et al. (1976, page 590-1)
wrote: \textquotedblleft A positive result is more likely, and a null result
is more informative, if the main comparison is of only 2 treatments, these
being as different as possible. \ldots\ [I]t is a mark of good trial design
that a null result, if it occurs, will be of interest.\textquotedblright%
\ \ This advice is equally relevant for observational studies, and it is part
of the reason that we prefer a conception in which there is a single treated
condition and a single control condition. \ Despite this, an investigator may
seek some reassurance that the study's conclusions cannot be undermined by the
possibility of two versions of treatment.

In that spirit, our analysis focuses on the main treatment-control comparison,
and subordinates the study of versions of treatment or versions of control.
\ In particular, the main treatment-control comparison is unaffected by the
exploration of versions of treatment --- the usual confidence interval for a
constant effect is reported --- despite controlling the family-wise error rate
in multiple comparisons that explore the possibility of versions of treatment
with different effects. \ Two confidence intervals are reported, the usual
interval for a constant effect and an interval designed to contain both
effects if the two versions differ. \ If the effect is constant, then both
intervals simultaneously cover that one effect with probability $\geq1-\alpha
$, but if there are two versions then the second interval covers both version
effects with probability $\geq1-\alpha$. \ The investigator always reports
both intervals, valid under different assumptions. \ This is an unusual type
of simultaneous inference: there is essentially one question, but there are
two sets of assumptions underlying the answer, so one question is answered
twice, as opposed to answering several different questions. \ There are
multiple assumptions rather than multiple hypotheses. \ The two intervals
together permit an investigator to report the conventional confidence interval
for a constant effect, without lengthening it for multiple testing, yet the
investigator also provides some information about whether the study's
conclusions depend on the absence of versions of treatment. \ The two
intervals may possibly disagree, say about whether no effect is plausible, and
if they do disagree then they demonstrate that the assumption about versions
of treatment is playing an important role in the interpretation of the
available data. \ Importantly, the method does not presume there is a single version of treatment by virtue of
failing to reject the null hypothesis that the two versions are equal. That is, it 
entertains the possibility that there are versions of treatment even when the convetional 
interval assuming a constant treatment effect is not empty. 

\section{ Possible versions of control in a study of football and dementia}

\label{secFootballExample}

There is evidence that severe repeated head trauma accelerates the on-set of
cognitive decline or dementia (Graves et al. 1990, Mortimer et al. 1991), with
specific concern about the risks faced by professional football players and
boxers (McKee et al. 2009, Lehman et al. 2012). \ It is unclear whether there
is also increased risk from playing football on a team in high school, but
there have been several recommendations against tackle football in high school
(Bachynski 2016, Miles and Prasad 2016). \ Does high school football
accelerate the on-set of cognitive decline?

A recent investigation used data from the Wisconsin Longitudinal Study,
comparing cognition and mental health measured at age 65 and 72, recorded in
2005 and 2011, of men who played football on a high school team in the mid 1950's
to male controls of similar age who did not play football (Deshpande et al. 2017). 
\ Following the practice in clinical trials, and as is recommended 
for observational studies by Rubin (2007), the design and protocol for this 
study were published on-line after
matching was completed but before outcomes were examined (Deshpande et al.
2016, \textsf{arXiv preprint arXiv:1607.01756}). \ The small number of people
who engaged in sports other than football with high incidences of head trauma
such as soccer, hockey, and wrestling were excluded from both the football and
control groups. \ One outcome was the 0-10 score on a ten item delayed word
recall (DWR) test at ages 65 and 72. \ The delayed word recall test was
designed as an inexpensive measure of memory loss associated with dementia;
see Knopman and Ryberg (1989). \ In this test, a person is asked to remember a
list of words that is then read to the person. \ Attention then shifts to
another activity, and after a delay, the person is asked to recall as many
words from the list as possible. \ The DWR score is the number of words
remembered. \ On average, in the Wisconsin Longitudinal Study, performance on
the delayed word recall test declined by half a word from age 65 to age 72.
\ It is useful to keep that half-word, 7-year decline in mind when thinking
about the magnitude of the effect of playing football.

A comparison of football players to all controls is natural, and might be
conducted without second thought. Among the controls, however, some played a
non-collision sport like baseball or track while others played no sports at
all. An investigator might reasonably seek reassurance that this natural
comparison has not oversimplified these two version of \textquotedblleft not
playing football.\textquotedblright\ At the same time, the investigator does
not want to sacrifice power to detect a constant treatment effect 
in the main comparison en route to obtaining this
reassurance by subdividing the data into many slivers of reduced sample size
and correcting for multiple comparisons. The method we propose achieves both
of these objectives.

Our question concerns the effects of high school football. \ It is important
to distinguish this question from questions about the effects of severe head
trauma in general. \ It is at least conceivable that high school football is
comparatively harmless, while severe head trauma is not, simply because severe
head trauma is not common in high school football, and the benefits of
exercise for all football players offset the harm of severe but rare trauma.
\ Conversely, severe head trauma from automotive or other accidents may be
difficult to prevent, but if high school football had grave consequences, then
it could simply be banned, in the same way that most high schools do not have
boxing teams. \ We ask about the effects of playing football in high school on
subsequent cognitive function.

The Wisconsin Longitudinal Study describes a specific piece of the US over a
specific period of time, and caution is advised about extrapolating its
conclusions to other times and places. \ High School football may have changed
since the 1950's, and the demographic composition of Wisconsin in the 1950's
is not the demographic composition of the US. \ The Wisconsin Longitudinal
Study is primarily a sequence of surveys, and it is impossible to use it to
investigate questions not asked in those surveys. \ For instance, we cannot
identify high school students who went on to play professional football, but
we suspect they were few in number. \ Because many young people play high
school football, the safety of high school football is an important question
apart from the safety of professional football.

\section{Full matching of football players and controls}

We matched the 591 male football players to all 1,290 male controls who did
not play football and did not play a contact sport. The match controlled for
several factors that may affect later-life cognition, including the student's
IQ score in high school, their high school rank-in-class recorded as a
percent, planned years of future education, as well as binary indicators of
whether teachers rated him as an exceptional student, and whether his teachers
and parents encouraged him to pursue a college education. We also accounted
for aspects of family background like parental income and education.

The match was a \textquotedblleft full match,\textquotedblright\ meaning that
a matched set could contain one football player and one or more controls, or
else one control and one or more football players. \ A full match is the form
of an optimal stratification in the sense that people in the same stratum are
as similar as possible subject to the requirement that every stratum contain
at least one treated subject and one control; see Rosenbaum (1991). \ Although
the proof of this claim requires some attention to detail, the key idea is
simple: if a matched set contained two treated subjects and two controls, it
could be subdivided into two matched sets that are at least as close on
covariates and are typically closer. \ 
\ See Hansen and Klopfer (2006) for an algorithm for optimal full
matching, Hansen (2007) for software, and Hansen (2004) and Stuart and Green
(2008) for applications. \ The match was constructed using Hansen's
\texttt{optmatch} package in \texttt{R} with the ratio of
controls to treated units constrained between 1:6 and 6:1 to avoid excessively
large matched sets.

In a full match, there are $I$ matched sets, $i=1,\ldots,I$ and $n_{i}$
individuals, $j=1,\ldots,n_{i}$, in set $i$. \ If individual $ij$ played on a
football team in high school, write $Z_{ij}=1$; otherwise, write $Z_{ij}=0$.
\ The number of football players in set $i$ is $m_{i}=\sum_{j=1}^{n_{i}}%
Z_{ij}$, the total number of individuals is $N=\sum_{i=1}^{I}n_{i}$, and the
total number of football players is $M=\sum_{i=1}^{I}m_{i}$. \ In a full
match, $\min\left(  m_{i},\,n_{i}-m_{i}\right)  =1$ for every $i$.

To explore versions of treatment, we constructed three matched samples. \ Each
sample used all $M=591$ football players. \ The first matched sample used all
controls, that is, every male who played neither football nor another contact
sport. \ The second matched sample used only controls who did not play any
sport. \ The third matched sample used controls who played a non-collision
sport, such as baseball. In each match, controls and football players belong\
to at most one matched set.\ Table
\ref{tabCounts}
describes the structure of the three matched samples, giving the frequency of
sets of size $\left(  m_{i},\,n_{i}-m_{i}\right)  $, as well as the number of
sets, $I$, the number of individuals, $N$, and the number of football players,
$M$. \ Obviously, the samples overlap extensively, because they all use all
$M=591$ football players and no controls were discarded in forming the optimal
full matchings; however, the three matches differ in structure,
partly because there were only $N-M=975-591=384$ controls who played a
non-collision sport in the third match. \ In all three matches, adequate covariate 
balance was achieved with nearly all standardized differences in baseline covariates
between football players and controls less than 0.2. Details of very similar matches
can be found in Deshpande et al. (2017).

\begin{table}[ht]
\caption{Distribution of matched set sizes, ($m_{i}$, $n_{i}-m_{i}%
$),  in three full matches.  A 2-1 set contains two treated
individuals and one control, while a 1-2 set contains one treated individual and two controls. There
are $I$ matched sets, containing a total of $N$ individuals, and each match includes all $M=591$
football players.}.\label{tabCounts}
\centering\begin{tabular}{l | rrrrrrrr | rrr}
\hline Comparison &  \multicolumn{8}{|c}{ (Treated Count)-(Control-Count) }
&  \multicolumn{3}{|c}{ Totals } \\ \hline
& 3-1 & 2-1 & 1-1 & 1-2 & 1-3 & 1-4 & 1-5 & 1-6   & $I$ &  $N$ & $M$ \\ \hline
Football vs. Control & 0 & 0 &  401 &  32 &  26 &  14 &  17 & 101  & 591  & 1881 & 591 \\
Football vs. No sport & 70 & 6 & 240 &  29 &  15 &  10 &   3 &  72 & 445 & 1497 & 591 \\
Football vs. Other sport & 90 & 43 & 227 &   3 &   2 &   3 & 0 & 0 & 368 & 975 & 591 \\
\hline\end{tabular}
\end{table}%

In studying the effects of a treatment --- here, high school football --- it
is typically inappropriate to adjust for events subsequent to the start of
treatment, as this may introduce bias even where none existed prior to
adjustments, because part of the treatment effect may be removed (Rosenbaum
1984). However, there are certain adult health outcomes that may be different between
football players and controls due to disparities in unmeasured baseline health
characteristics rather than an effect of playing football. This may threaten the validity 
of our study if these baseline health characteristics also play a role in later-life cognitive health. 
Comparing these health outcomes may can be used, at least partially, to assess the comparability of the baseline 
health of the comparison groups. \ We checked on the health status of football players
and matched controls at age 65 using the Mantel-Haenszel procedure, failing to
find a difference significant at the 0.05 level for \textquotedblleft ever had
high blood pressure,\textquotedblright\ \textquotedblleft ever had
diabetes,\textquotedblright\ and \textquotedblleft ever had heart
problems\textquotedblright. \ Football players were more likely to report that
they had \textquotedblleft ever had a stroke,\textquotedblright\ with a
$P$-value of 0.03, and a 95\% confidence interval for the odds ratio of
$\left[  1.09,\,3.21\right]  $. \ Extensive comparisons of this kind are
reported in Deshpande et al. (2017).

\section{Review of randomization inference without versions of treatment}

\label{secInferenceWithoutVersions}

If there were a single version of treatment or control, then individual $ij$
would have two potential delayed word recall scores, $Y_{ij}(1)$ if he played
football and $Y_{ij}(0)$ if he did not, where we observe only one of these,
namely $Y_{ij}=Z_{ij}\,Y_{ij}(1)+\left(  1-Z_{ij}\right)  \,Y_{ij}(0)$, and the
effect caused by playing football, namely $\delta_{ij}=Y_{ij}(1)-Y_{ij}(0)$, is
not observed for any individual; see Neyman (1923) and Rubin (1974).
\ Fisher's (1935) sharp null hypothesis of no effect says $H_{0}%
:Y_{ij}(1)=Y_{ij}(0)$, $i=1,\ldots,I$, $j=1,\ldots,n_{i}$, which we henceforth
abbreviate as $H_{0}:Y_{ij}(1)=Y_{ij}(0)$, $\forall i,j$ or as $H_{0}:\delta
_{ij}=0$, $\forall i,j$. \ The treatment has an additive constant effect if
there exists some constant $\tau$ such that $\delta_{ij}=Y_{ij}(1)-r_{Cij}=\tau
$, $\forall i,j$. \ The hypothesis $H_{\tau_{0}}$ specifies a particular
numerical value $\tau_{0}$ for $\tau$ and asserts $H_{\tau_{0}}:\delta
_{ij}=\tau_{0}$, $\forall i,j$, and it is manifested in the observable
distribution of $Y_{ij}$ by a within-set shift in the distribution of $Y_{ij}$
by $\tau_{0}$. \ If $H_{\tau_{0}}$ were true, then $Y_{ij}-\tau_{0}%
\,Z_{ij}=Y_{ij}(0)$ would satisfy Fisher's hypothesis of no effect, $H_{0}$, and
it is commonplace to test $H_{\tau_{0}}$ by replacing $Y_{ij}$ by $Y_{ij}%
-\tau_{0}\,Z_{ij}$ and testing $H_{0}$.

Until \S \ref{secSensitivity}, we restrict attention to random assignment of
treatments within matched sets; however, \S \ref{secSensitivity} considers
sensitivity of inferences to departures from this assumption. \ Of course,
people do not decide to play football at random, so \S \ref{secSensitivity} is
closer to reality than random assignment. \ Fisher (1935), Pitman and Welch
(1937) used the randomization distribution of the mean difference to test
Fisher's $H_{0}$, and we follow this approach with the short-tailed delayed
word recall scores (DWR), only briefly comparing the mean to a robust
$M$-statistic. \ The mean is one $M$-statistic, but not a robust one.
\ Because the matched sets are of unequal sizes, $\left(  m\,_{i}%
,\,n_{i}-m_{i}\right)  $, we compute the treated-minus-control mean difference
in DWR scores within each set $i$ and combine them with efficient weights
based on the matched set sizes; see Rosenbaum (2007b, \S 4.1) for discussion
of these weights, which are implemented in the \texttt{senfm} function of the
\texttt{sensitivityfull} package in \texttt{R} with option \texttt{trim=Inf}.
\ 

As is always true, a $1-\alpha$ confidence interval $\mathcal{I}_{c}$ for
$\tau$ is formed by inverting a level-$\alpha$ test, so $\mathcal{I}_{c}$ is
the shortest interval of values of $\tau_{0}$ not rejected by the test; see
Lehmann and Romano (2005, \S 3) for general discussion. \ Typically, a
two-sided confidence interval is the intersection of two one-sided
$1-\alpha/2$ confidence intervals; see Shaffer (1974).

Ignoring versions of treatment, using the first match in Table
\ref{tabCounts}%
, and assuming that treatments are randomly assigned within matched sets, we
obtain a randomization-based 95\% confidence interval of $\left[
-0.308,\,0.099\right]  $ for $\tau$, that is, for a constant effect of playing
football on the number of words remembered in the delayed word recall test.
\ Because this confidence interval includes zero, the hypothesis of no effect
is not rejected at the 0.05 level. \ Because this confidence interval excludes
all $\tau$ with $\left\vert \tau\right\vert \geq1/3$, constant effects of
$\pm1/3$ word remembered have been rejected as too large. \ It is important
that \textquotedblleft no effect\textquotedblright\ is plausible, but equally
important that large effects, positive or negative, are implausible values for
a constant effect, $\tau$. \ Our goal is to avoid lengthening this interval
for $\tau$ as we explore possible versions of the control, while controlling
the family-wise error rate at $\alpha$, conventionally $\alpha=0.05$. \ This
simultaneous inference is possible if the exploration of versions of treatment
takes a specific form.

Incidentally, had we built the confidence interval for $\tau$ using the
default $M$-estimate in the \texttt{senfm} function, rather than the mean with
option \texttt{trim=Inf}, then the 95\% randomization interval for $\tau$
would have been $\left[  -0.315,0.096\right]  $. The default $M$-estimate in \texttt{senfm}\ 
corresponds to Huber's $\psi$-function, i.e., $\psi(y) = y$ for $|y|\le 1$ and
$\psi(y)=\text{sign}(y)$ for $|y| > 1$. Generally, use of robust
procedures is advisable, but we do not do so in this example to simplify its
presentation, as the robust procedures give similar answers in this
short-tailed example.

\section{Inference with versions of treatment}

\subsection{Structure of the problem}

\label{ssStructure}

With two versions of control, say \textquotedblleft playing no
sport\textquotedblright\ and \textquotedblleft playing a non-collision
sport\textquotedblright\ like baseball, each person has two potential control
responses, $Y_{ij}(0,a)$ and $Y_{ij}(0,b)$, and hence
two treatment effects, $\delta_{ij}^{{a}}=Y_{ij}(1)-Y_{ij}(0,a)$
and $\delta_{ij}^{{b}}=Y_{ij}(1)-Y_{ij}(0,b)$.\ If
$Y_{ij}(0,a)=Y_{ij}(0,b)$, $\forall i,j$, then the two
versions of control yield the same effects, $\delta_{ij}^{{a}}%
=\delta_{ij}^{{b}}$, and so the versions need not be distinguished.\
This notation for potential outcome under versions of control follows 
Vanderweele and Hernan (2013) where potential outcomes are fixed 
by {\it both} treatment and version. If there are versions, the implied
randomization distribution will have three arms when matched sets include
both versions of control, which might complicate inference. However, we 
only use the matching with sets including both versions of control to 
conduct inference under the assumption that the versions are irrelevant,
and thus the three arm randomization distribution collapses to the simpler
two arm design.

Consider the two null hypotheses about additive effects for the two versions
of control, $H_{\tau_{0}}^{{a}}:\delta_{ij}^{{a}}=\tau_{0}$,
$\forall i,j$ and $H_{\tau_{0}}^{^{{b}}}:\delta_{ij}%
^{{b}}=\tau_{0}$, $\forall i,j$. \ Here, $H_{\tau_{0}}^{{a}%
}$ might be true when $H_{\tau_{0}}^{^{{b}}}$ is false, or
conversely. \ Define $H_{\tau_{0}}$ to be the hypothesis that both
$H_{\tau_{0}}^{{a}}$ and $H_{\tau_{0}}^{^{{b}}}$ are true,
that is, $H_{\tau_{0}}:\delta_{ij}^{{a}}=\delta_{ij}^{{b}%
}=\tau_{0}$, $\forall i,j$, so the two versions of control yield the same
effect $\tau_{0}$ and need not be distinguished. \ By the definition of
$H_{\tau_{0}}$, if either $H_{\tau_{0}}^{{a}}$ or $H_{\tau_{0}%
}^{^{{b}}}$ is false, then $H_{\tau_{0}}$ is false; that is, if
there are two versions of treatment or control with different effects, then
there is not a constant effect.

It is straightforward to test $H_{\tau_{0}}^{{a}}$ or $H_{\tau_{0}%
}^{^{{b}}}$ using the methods in
\S \ref{secInferenceWithoutVersions} simply by restricting attention to
controls of one type or the other. \ These tests will be based on a smaller
sample size than the test in \S \ref{secInferenceWithoutVersions} because not
all of the controls are used. \ Moreover, if $H_{\tau_{0}}$, $H_{\tau_{0}%
}^{{a}}$ and $H_{\tau_{0}}^{^{{b}}}$ are each tested at
level $\alpha$, then the chance of at least one false rejection would
typically exceed $\alpha$ unless something is done to control the family-wise
error rate. \ Understandably, an investigator would like to avoid weakening
the inference about $H_{\tau_{0}}$ by virtue of considering $H_{\tau_{0}%
}^{{a}}$ and $H_{\tau_{0}}^{^{{b}}}$, and the question is
how to achieve the investigator's goals.

Suppose there are two versions of a constant additive treatment effect,
$\delta_{ij}^a = \tau^a$ and $\delta_{ij}^b=\tau^b$ for every $i,j$. Let $\tau_{\mathrm{\min}}=\min\left(  \tau^{{a}},\tau^{{b}%
}\right)  $ and $\tau_{\mathrm{\max}}=\max\left(  \tau^{{a}}%
,\tau^{{b}}\right)  $. \ If $\tau^{{a}}=\tau^{{b
}}=\tau$, then $\tau_{\mathrm{\min}}=\tau$ and $\tau_{\mathrm{\max}%
}=\tau$, so the versions do not matter. \ Our approach in
\S \ref{ssIntervalEstimates} is to build two confidence intervals, one
interval for $\tau$ and another interval designed to contain $\left[
\tau_{\mathrm{\min}},\tau_{\mathrm{\max}}\right]  $. \ If there is no need to
consider versions of treatment or control because $Y_{ij}(0,a)=Y_{ij}(0,b)$, 
implying that $\tau^{{a}}%
=\tau^{{b}}=\tau$, then with probability at least $1-\alpha$, both
intervals simultaneously cover the true $\tau$. \ If $\tau^{{a}}\neq
\tau^{{b}}\,$, then $H_{\tau_{0}}$ is false for every $\tau_{0}$,
but with probability at least $1-\alpha$ the second interval covers the
interval $\left[  \tau_{\mathrm{\min}},\tau_{\mathrm{\max}}\right]  $.
\ Moreover, the first interval for $\tau$ is the interval reported in
\S \ref{secInferenceWithoutVersions} ignoring versions of treatment, so the
investigator has ensured that under either the assumption of a constant treatment effect
or the assumption of versions, the two intervals he reports control the family-wise error rate at
$\alpha$, while paying no additional price in power to detect a constant effect for
consideration of versions of treatment. The inference is simultaneous in that the two
intervals provide the correct coverage for $\tau$ under the assumption of a constant effect 
and also the correct coverage for $\tau^a$ and $\tau^b$ under the presence of versions.

\subsection{Inference when there may or may not be two versions of treatment}
\label{ssIntervalEstimates}
The theory in this section is derived for one-sided intervals but, as we will see
shortly, is easily extended to the two-sided $1-\alpha$ intervals described in the 
previous section. There is a valid, one-sided $P$-value, say $P_{\tau_{0}}^{a}$, testing
$H_{\tau_{0}}^{{a}}$ against $\tau^{a}>\tau_{0}$, so that
$\Pr\left(  P_{\tau_{0}}^{a}\leq\alpha\right)  \leq\alpha$ if
$H_{\tau_{0}}^{{a}}$ is true. \ In parallel, there is a valid one-sided
$P$-value $P_{\tau_{0}}^{{b}}$, testing $H_{\tau_{0}}%
^{^{{b}}}$ against $\tau^{b}>\tau_{0}$, and a valid
one-sided $P$-value, $P_{\tau_{0}}$, testing $H_{\tau_{0}}$ against $\tau
>\tau_{0}$. \ With a slight abuse of notation, write the probability that a
random interval \ $\mathcal{I}$ contains a fixed real number $\tau$ as
$\Pr\left(  \mathcal{I}\supseteq\tau\right)  $. \ Under the assumption,
perhaps incorrect, that the there is a single version of the treatment,
$\tau^{{a}}=\tau^{{b}}=\tau$, let $\mathcal{I}^-_{c}$ be the
usual one-sided $1-\alpha$ confidence interval for $\tau$ formed by inverting
the test of $H_{\tau_{0}}$, so $\mathcal{I}^-_{c}$ is the smallest set of the
form $\left[  \widetilde{\tau},\,\infty\right)  $ containing $\left\{
\tau_{0}:P_{\tau_{0}}>\alpha\right\}  $. \ If there are no versions of
treatment, so $\tau^{{a}}=\tau^{{b}}=\tau$ for some $\tau$,
then $\Pr\left(  \mathcal{I}^-_{c}\supseteq\tau\right)  \geq\alpha$ by the
familiar duality of tests and confidence intervals; see Lehmann and Romano
(2005, Chapter 3). \ The investigator would like to report this standard
interval $\mathcal{I}^-_{c}$\ for a constant effect, without lengthening it for
multiple testing, yet would like to also say something about the possibility
that there are versions of treatment with $\tau^{{a}}\neq\tau
^{{b}}$. \ Of course, if there are versions of treatment with
$\tau^{{a}}\neq\tau^{{b}}$ then $H_{\tau_{0}}$ is false for
every $\tau_{0}$ and there is no true value of $\tau$ for $\mathcal{I}^-_{c}$ to
contain or omit.

The smallest set of the form $\left[  \widetilde{\tau},\,\infty\right)  $
containing $\left\{  \tau_{0}:P_{\tau_{0}}>\alpha\;\mathrm{or\;}P_{\tau_{0}%
}^{{a}}>\alpha\;\mathrm{or\;}P_{\tau_{0}}^{{b}}%
>\alpha\right\}  $ will be denoted $\mathcal{I}^-_{v}$. \ Of course,
$\mathcal{I}^-_{v}\supseteq\mathcal{I}^-_{c}$. \

The investigator does not know whether or not there are two versions of
treatment, whether or not $\tau^{{a}}=\tau^{{b}}$. \ The
investigator would like to make two inferences appropriate for the two
situations, $\tau^{{a}}=\tau^{{b}}$ or $\tau^{{a}}%
\neq\tau^{{b}}$. \ The investigator would like to make an
inference appropriate to this state of ignorance. \ The investigator says
\textquotedblleft I do not know whether there are two versions of treatment,
whether or not $\tau^{{a}}=\tau^{{b}}$; however, (i) if
there are not versions of treatment so that $\tau^{{a}}=\tau
^{{b}}=\tau$, then $\mathcal{I}^-_{c}\supseteq\tau$, and (ii)
whether or not there are two versions of treatment, even if $\tau^{{a}%
}\neq\tau^{{b}}$, then $\mathcal{I}^-_{v}\supseteq\tau
_{\mathrm{\min}}$; moreover, this method produces two true hypothetical
statements with probability at least $1-\alpha$.\textquotedblright%
\ \ Statement (ii) cost nothing, in the sense that $\mathcal{I}^-_{c}$ is the
usual one-sided confidence interval for $\tau$ assuming there are not versions
of treatment, yet both statements hold jointly without multiplicity
correction. \ This is established in the following proposition.

\begin{proposition}
\label{PropMain} (i) If there is only one version of treatment, $\tau
^{{a}}=\tau^{{b}}=\tau$, then $\Pr\left(  \mathcal{I}^-%
_{v}\supseteq\mathcal{I}^-_{c}\supseteq\tau\right)  \geq 1-\alpha$. \ (ii) In any
event, whether there are two versions of treatment, $\tau^{{a}}\neq
\tau^{{b}}$, or only a single version, $\tau^{{a}}%
=\tau^{{b}}=\tau$, we have $\Pr\left(  \mathcal{I}^-_{v}%
\supseteq\tau_{\mathrm{\min}}\right)  \geq 1-\alpha$.
\end{proposition}

\begin{proof}
By the definitions of $\mathcal{I}^-_{v}$ and $\mathcal{I}^-_{c}$, we have
$\mathcal{I}^-_{v}\supseteq\mathcal{I}^-_{c}$. \ Then (i) follows because, if
there are not versions of treatment, $\tau^{{a}}=\tau^{{b}%
}=\tau$, then $\mathcal{I}^-_{c}$ is a $1-\alpha$ confidence interval for $\tau$
and $\Pr\left(  \mathcal{I}^-_{v}\supseteq\mathcal{I}^-_{c}\supseteq\tau\right)
\geq 1-\alpha$. \ If there are not versions of treatment, $\tau^{{a}}%
=\tau^{{b}}=\tau$, then $\tau_{\mathrm{\min}}=\tau$, so
$\Pr\left(  \mathcal{I}^-_{v}\supseteq\tau_{\mathrm{\min}}\right)  \geq 1-\alpha$,
as required for (ii). \ So suppose there are two versions of treatment. \ If
$\tau^{a}=\tau_{\mathrm{\min}}<\tau_{\mathrm{\max}}=\tau^{b}$,
then $\tau_{\mathrm{\min}}\notin\mathcal{I}^-_{v}$ implies $P_{\tau^{a}%
}^{{a}}\leq\alpha$ which occurs with probability at most $\alpha$. \ If
$\tau^{b}=\tau_{\mathrm{\min}}<\tau_{\mathrm{\max}}=\tau^{a}$,
then $\tau_{\mathrm{\min}}\notin\mathcal{I}^-_{v}$ implies $P_{\tau
^{b}}^{b}\leq\alpha$ which occurs with probability at
most $\alpha$. \ So in all three cases, $\tau^{{a}}=\tau^{{b
}}$ or $\tau^{a}<\tau^{b}$ or $\tau^{a}%
>\tau^{b}$, we have $\Pr\left(  \mathcal{I}^-_{v}\supseteq
\tau_{\mathrm{\min}}\right)  \geq 1-\alpha$, proving (ii).
\end{proof}

By a parallel argument, we obtain analogous $1-\alpha$ upper intervals,
$\mathcal{I}_{c}^{+}$ and $\mathcal{I}_{v}^{+}$, of the form $\left(
-\infty,\widetilde{\tau}\right)  $ for $\tau$ if $\tau^{{a}}%
=\tau^{{b}}=\tau$ or without restrictions for $\tau_{\mathrm{\max
}}$. \ Taking the intersections, $\mathcal{I}_{c}=\mathcal{I}^-_{c}\cap\mathcal{I}_{c}^{+}$ and
$\mathcal{I}_{v}=\mathcal{I}^-_{v}\cap\mathcal{I}_{v}^{+}$, of two one-sided $1-\alpha/2$
intervals yields analogous two-sided $1-\alpha$ intervals for $\tau$ if
$\tau^{{a}}=\tau^{{b}}=\tau$ or without restrictions for the
interval $\left[  \tau_{\mathrm{\min}},\tau_{\mathrm{\max}}\right]  $. In most
cases, $\mathcal{I}_v$ can be constructed by taking the union of $\mathcal{I}_c$
and the two two-sided intervals constructed from the matched sets using each separate
version of control. When this union is disjoint, $\mathcal{I}_v$ is the shortest
interval that contains all three intervals.

In case (ii), the proof above that $\Pr\left(  \mathcal{I}_{v}\supseteq
\tau_{\mathrm{\min}}\right)  \geq 1-\alpha$ is similar to, but not quite
identical to, results in Lehmann (1952), Berger (1982) and Laska and Meisner
(1989). \ These authors proposed tests that would invert to yield as a
confidence interval the shortest interval $\mathcal{I}_{\ast}$ containing
$\left\{  \tau_{0}:P_{\tau_{0}}^{{a}}>\alpha\;\mathrm{or\;}P_{\tau_{0}%
}^{{b}}>\alpha\right\}  $, whereas $\mathcal{I}_{v}$ is the
shortest interval containing $\left\{  \tau_{0}:P_{\tau_{0}}>\alpha
\;\mathrm{or\;}P_{\tau_{0}}^{{a}}>\alpha\;\mathrm{or\;}P_{\tau_{0}%
}^{{b}}>\alpha\right\}  $, thereby ensuring $\mathcal{I}%
_{v}\supseteq\mathcal{I}_{c}$. \ Of course, $\mathcal{I}_{v}\supseteq
\mathcal{I}_{\ast}$, but unlike $\mathcal{I}_{\ast}$, our method ensures that
$\mathcal{I}_{v}$ and $\mathcal{I}_{c}$ both simultaneously cover
$\tau^{a}=\tau^{b}=\tau$ at rate $1-\alpha$ when there is
actually only a single version of treatment. \ Because $\mathcal{I}_{c}$ is
built using all of the data and under stronger assumptions, it is unlikely
that $\mathcal{I}_{\ast}$ will be much shorter than $\mathcal{I}_{v}$;
however, this logical possibility is the price for reporting the usual
interval, $\mathcal{I}_{c}$, without multiplicity correction.

Why not report a single robust interval like $\mathcal{I}_*$ instead of two
intervals, $\mathcal{I}_c$ and $\mathcal{I}_v$, that have slightly more nuanced
coverage properties? A simple example may help illustrate the advantage of reporting
$\mathcal{I}_c$ and $\mathcal{I}_v$ over $\mathcal{I}_*$. Suppose that $\mathcal{I}_*$
and $\mathcal{I}_v$ both contain zero but $\mathcal{I}_c$ does not. If we choose to
report $\mathcal{I}_*$ then we have little additional information to determine why 
$\mathcal{I}_*$ contains zero -- is Fisher's sharp null true or is the smaller of the
the two versions of effect very close to zero? Or are there versions of the effect 
with different signs? However, if we report $\mathcal{I}_c$ and $\mathcal{I}_v$ we 
have evidence against Fisher's sharp null, narrowing the plausible explanations for
why zero is contained in $\mathcal{I}_v$, which we've noted will tend to be similar to $\mathcal{I}_*$. 

\subsection{Interval estimates in the football study}

The upper third of Figure 1, marked $\Gamma=1$, shows 95\% intervals for the
football study, assuming that treatments are randomly assigned within matched
sets. \ First, there are the three conventional intervals for $\tau$,
$\tau^{a}$, and $\tau^{b}$, corresponding to the three
comparisons in Table
\ref{tabCounts}%
. \ Each of these intervals is a 95\% confidence interval on its own, but each
one runs a 5\% chance of error, so the chance that at least one interval fails
to cover its corresponding parameter is greater than 5\%. Obviously, we could
make the three intervals longer, say using the Bonferroni inequality, so that
the simultaneous coverage is 95\%, but many investigators would find this
unattractive because it would reduce the power of the conventional, primary
analysis focused on $\tau$ that uses all of the controls; that is, it would
make the first interval longer.

\begin{figure}[ht!]
    \centering
    \includegraphics[scale=0.8]{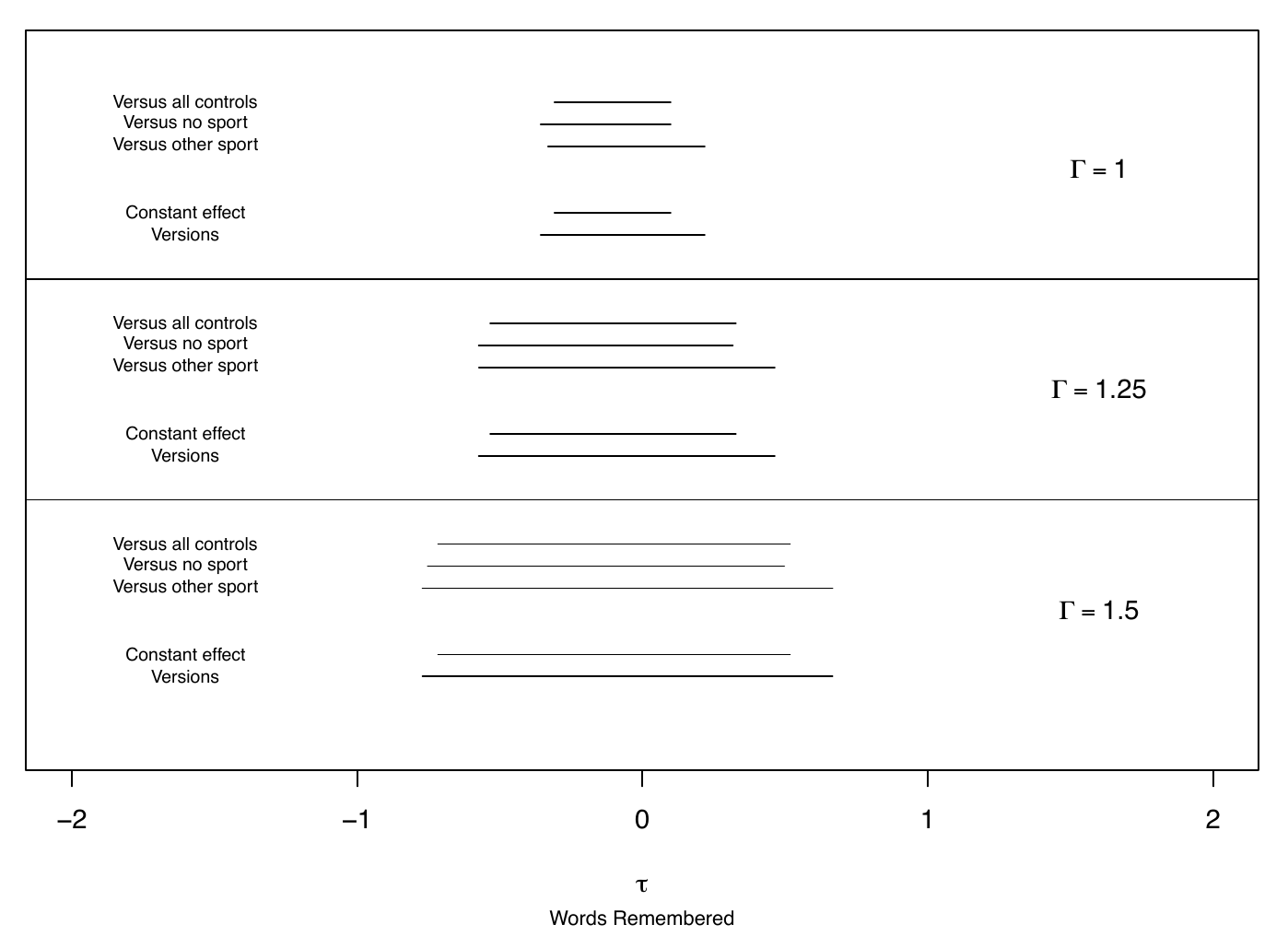}
    \caption{Comparison of interval estimates for the effect of high school football 
    on the delayed word recall score. The intervals for $\Gamma=1$ assume that there 
    is no bias from unmeasured covariates, while $\Gamma > 1$ permits unmeasured biases
    of unknown form but limited magnitude. The top three intervals for ``all controls'',
    ``no sport'', and ``other sport'' are conventional confidence intervals lacking
    simultaneous coverage. The bottom two intervals are $\mathcal{I}_c$ and $\mathcal{I}_v$.
    Noitce that the ``all controls'' interval equals $\mathcal{I}_c$ and the union of the first
    three intervals equals $\mathcal{I}_v$.}
\end{figure}

In contrast, the intervals $\mathcal{I}_{c}$ and $\mathcal{I}_{v}$ in Figure 1
have simultaneous coverage of 95\% in the sense of Proposition \ref{PropMain}.
\ Notably, $\mathcal{I}_{c}=\left[  -0.308,\,0.099\right]  $ is the interval
for $\tau$ from \S \ref{secInferenceWithoutVersions}, so consideration of
$\mathcal{I}_{v}$ has not reduced power to detect a constant effect.
\ The versions $\mathcal{I}_{v}=$ $\left[  -0.357,\,0.219\right]  $ is
slightly longer than $\mathcal{I}_{c}$, but both intervals are compatible with
no effect and both intervals are quite incompatible with an effect of half a
word, $\pm0.5$. \ For comparison, recall from \S \ref{secFootballExample} that
average performance on the delayed word recall test declined by half a word
from age 65 to age 72. \ In Figure 1, the 95\% interval for \textquotedblleft
all controls\textquotedblright\ equals $\mathcal{I}_{c}$, while $\mathcal{I}%
_{v}$ is the union of the three intervals for \textquotedblleft all
controls\textquotedblright, \textquotedblleft controls who played no
sport\textquotedblright, and \textquotedblleft controls who played another
sport\textquotedblright.

\section{Comparison to conventional approaches to multiple versions: F-tests and Bonferroni correction}

The method described in \S 5.2 makes an important trade-off: it prioritizes
the primary comparison against all controls under the assumption of a constant treatment
effect, allowing us to report the corresponding confidence interval with no correction,
in exchange for the ability to distinguish between versions if they do, in fact, exists. In 
other words, our method emphasizes detection of non-zero, constant treatment effects over 
detecting different versions of the effect. Conventional approaches to multiple comparisons,
are often less focused and are designed to detect a broader range of departures from the null.
For example, a simple Bonferroni correction places the primary comparison and the two versioned
comparisons on equal footing. If the investigator does not suspect a priori that a particular
alternative hypothesis is most likely, he may conduct an omnibus F-test, whose power is distributed
over a broad range of alternative hypotheses.

The researcher's scientifc aim should determine which method for multiple comparisons
is most appropriate. With this guidance in mind, we compare our method to the omnibus 
F-test and Bonferroni corrected intervals.
\newline

\noindent\textit{The omnibus F-test}

In exploratory analyses, the F-test can be a useful ``prelude to subsequent examinations of unplanned
contrasts" (Steiger, 2004). However, in many studies, the researcher will have a particular contrast 
in mind. Several authors have argued that the omnibus hypothesis in ANOVA studies be replaced with
hypotheses that focus on a substantive research question, often involving just a single
contrast (Rosenthal et al., 2000; Steiger, 2004). In the football study, we suspect a priori 
that versions are not terribly consequential and proceed first with our primary investigation of 
whether playing high school football accelerates the onset of cognitive decline. The hypotheses about
versions are secondary to our main inquiry and are treated as such in our method.

The F-test does not lend itself to effect size estimates, but we can compare it to our method 
by evaluating it's power to reject Fisher's sharp null hypothesis and the probability that $\mathcal{I}_c$
excludes $\tau=0$ under a variety of alternative hypotheses with vary degrees of ``versioning." If the 
primary goal of the study is less directed and detecting any departure from the null of no effect is of interest
and the researcher suspects that versions may play an important role, the F-test may be more suitable.
However, in a simulation study described in the Appendix, we find that when the versions 
differ by less than $\sim 30-40\%$ in magnitude and have the same sign our method has better power to 
reject Fisher's sharp null than the F-test. Being an omnibus test, the F-test power is much less sensitive
to the specific pattern of alternative hypothesis.
\newline

\noindent\textit{Bonferroni corrected intervals}

In Figure 2, we compare the intervals returned by our version method to the three intervals returned
by a Bonferroni correction (BC) to the comparison of the football players against all controls and each 
version of control. In the top panel, the BC interval using all controls is 22\% longer than $\mathcal{I}_c$.
In the bottom panel, $\mathcal{I}_v$ is only 3\% longer than the BC interval using only ``no sport'' controls
and 15\% \textit{shorter} than the BC interval using ``other sport" controls. In this particular case,
where versions appear to be relatively innocuous, the cost of lengthening the primary interval is significant 
for what apears to be little if any gain from investigating each version separately. Bonferroni correction may
be more appropriate when the investigator suspects the versions are important, and he may even replace the
primary comparison with a comparison of the two versions of control themselves. However, if the versions are 
only modestly different, for which our method is designed, it is unlikely that Bonferroni would have sufficient
power to detect such modest differences in the versions of control.

\begin{figure}[ht!]
    \centering
    \includegraphics[scale=0.6]{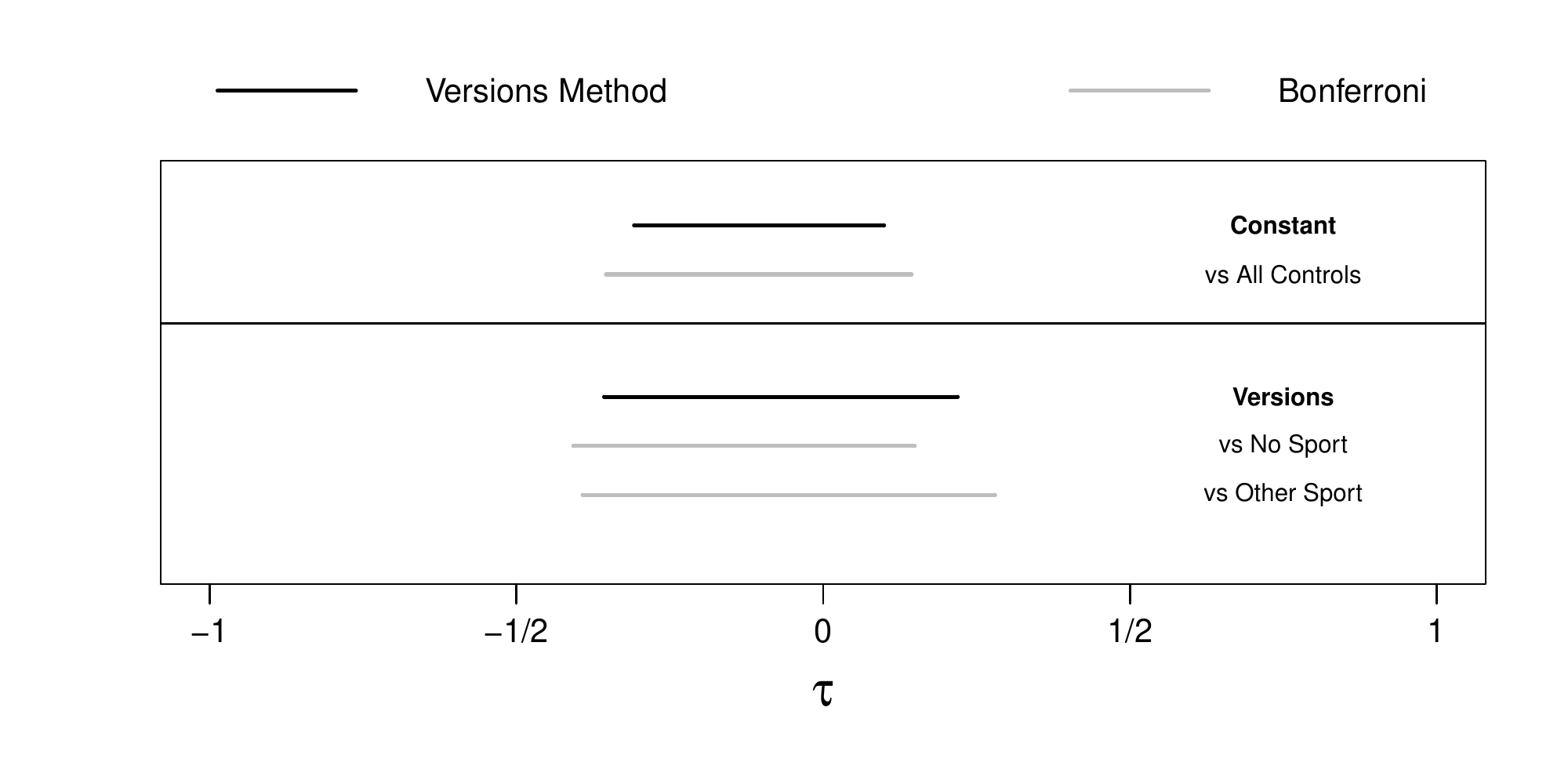}
    \caption{Comparison of $\mathcal{I}_c$ (top panel, dark line) and $\mathcal{I}_v$ (bottom panel, dark line)
    to Bonferroni corrected (BC) intervals comparing football players vs. ``all controls'' (top panel, light line), 
    vs. ``no sport'' controls (bottom panel, top-most light line), and vs. ``other sport'' controls 
    (bottom panel, bottom-most light line). Notice that $\mathcal{I}_v$ is very similar in length to the BC interval
    using ``no sport'' controls and is noticeably shorter than the BC interval using ``other sport'' controls.}
\end{figure}
\section{Versions, effect modification, and insufficient overlap}

In the study of the effects of playing high school football on cognitive decline, the two versions of control
and the football players have significant overlap in observed covariates. An anonymous reviewer suggested the
following situation. Suppose that athletes who did not play football and non-athletes who did not play football
differ noticeably on observed covariates such that a different set of football players were matched in the matchings 
using the different versions of control. If the treatment is heterogeneous, say it is modified by some observed
covariate that differs between the two versions of control, does ``this mean there are two versions of treatment
or that there are two groups of individuals that are involved in the two comparisons." In this setting, the potential 
for treatment heterogeneity is aliased with the potential for versions of the treatment effect. But does this matter?
The logic of \S 5.2 is agnostic to why there may be different treatment effects between versions. Thus, $\mathcal{I}_v$
can be interpreted as assessing how robust our primary analysis is to the existence of versions of treatment {\it or} to 
the existence of effect heterogeneity between the two groups defined by ``versions" of control.

\section{Sensitivity to departures from random assignment}

\label{secSensitivity}

So far, we have drawn inferences under the assumption that treatments are
randomly assigned within matched sets. \ In an observational study, this
assumption lacks support and is typically doubtful if not implausible. \ We
examine sensitivity to bias from nonrandom assignment by assuming that two
individuals with the same observed covariates may differ in their odds of
treatment by at most a factor of $\Gamma\geq1$ due to differences in
unobserved covariates; see Rosenbaum (2007b; 2017, \S 9). \ This yields
hypothesis tests that falsely reject a true null hypothesis with probability
at most $\alpha$ when the bias in treatment assignment is at most $\Gamma$.
\ Then $\Gamma$ is varied to display the magnitude of bias that would need to
be present to alter the conclusions of a study. \ How much bias, measured by
$\Gamma$, would need to be present to lead us to fail to reject the null
hypothesis of no effect of football when, in fact, football causes substantial harm?
In the current example where there is no evidence of a harmful effect of football,
we may ask a parallel question that is related to equivalence testing: 
how much bias would need to be present to mask a substantial true effect of football 
on memory? For example, an increase or decrease of a DWR score by at least one word.

Aids to interpreting values of $\Gamma$ are discussed by Rosenbaum and Silber
(2009) and Hsu and Small (2013). \ In particular, in a matched pair with
$n_{i}=2$, the value $\Gamma=1.25$ corresponds with an unobserved covariate
that doubles the odds of playing football and doubles the odds of a worse
memory score, while $\Gamma=1.5$ corresponds with an unobserved covariate that
doubles the odds of playing football and quadruples the odds of a worse memory
score; see Rosenbaum and Silber (2009) and Rosenbaum (2017, \S 9).
\ Proposition \ref{PropMain} applies to the intervals obtained from upper
bounds on $P$-values from sensitivity analyses, providing the bias in
treatment assignment is at most $\Gamma$.

Figure 1 shows the expansion of $\mathcal{I}_{c}$ and $\mathcal{I}_{v}$ as
$\Gamma$ increases from $\Gamma=1$ for randomization inferences to
$\Gamma=1.25$ and $\Gamma=1.5$. \ For $\Gamma=1.25$, the intervals are
$\mathcal{I}_{c}=\left[  -0.534,\,\,0.328\right]  $ and $\mathcal{I}%
_{v}=\left[  -0.574,\,0.464\right]  $. \ For $\Gamma=1.5$, the intervals are
$\mathcal{I}_{c}=\left[  -0.716,\,\,0.517\right]  $ and $\mathcal{I}%
_{v}=\left[  -0.771,\,0.666\right]  $. \ A bias of $\Gamma=1.5$ together with
two versions of not playing football would be insufficient to mask an effect
of one word on the memory test, $\pm1$. \ At $\Gamma=2$, not shown in Figure
1, effects of $\pm1$ word start to be included in the confidence intervals,
with $\mathcal{I}_{c}=\left[  -0.997,\,0.817\right]  $ and $\mathcal{I}%
_{v}=\left[  -1.082,\,0.986\right]  $. \ A bias of $\Gamma=2$ corresponds with
an unobserved covariate that triples the odds of playing football and
increases the odds of worse memory performance by five-fold.

In brief, there is no sign of an effect of high school football on memory
scores. \ Could the absence of any sign of an effect reflect a substantial
effect and bias in who plays football? \ To mask a true effect of $\pm1$ word,
an unobserved bias would have to be moderately large, $\Gamma=2$. \ Even
allowing for both moderate confounding due to unmeasured covariates and
versions of treatment, large effects of high school football on memory scores
are not consistent with the data.

\section{Discussion: Simultaneous inference about one question under different
assumptions}

Investigators sometimes candidly report two or more statistical analyses valid
under different assumptions. \ In the process, they often lose the several
advantages of a single, simple, primary analysis, that is, a single analysis
with high power against a non-zero constant effect 
because it uses everyone and avoids needed corrections for
multiple testing when several statistical tests are performed. With less
candor, investigators sometimes perform several analyses and report some but
not all analyses, a perhaps common practice that no one would publicly advocate.

Versions of treatment arise in observational studies when treatment or control
conditions found in available data may not be uniform, as they would be in a
tightly controlled experiment. \ The investigator would like follow the
practice of clinical trials and report a single, primary analysis using
everyone without multiplicity correction. \ Nonetheless, the investigator
would like to speak to the possibility that there are versions of treatment or
control conditions. \ The proposed method always reports two interval
estimates. \ The first, shorter interval, $\mathcal{I}_{c}$, is precisely the
interval that would be reported in a single primary analysis without versions
of treatment. \ The second longer interval, $\mathcal{I}_{v}$, attempts to
cover both treatment effects if there are two versions of treatment or two
versions of control. \ If there is, in fact, only a single treatment effect,
the same for both versions, then the probability that both $\mathcal{I}_{c}$
and $\mathcal{I}_{v}$ simultaneously cover that one effect is the stated rate
of $1-\alpha$. \ If there are, in fact, two treatment effects that differ with
the two versions, then the second interval, $\mathcal{I}_{v}$, covers both
effects with the stated rate of $1-\alpha$. \ In that sense, the added
information provided by reporting two intervals, $\mathcal{I}_{c}\ $and
$\mathcal{I}_{v}$, is free: the interval $\mathcal{I}_{c}$ is clarified but
not lengthened by examining $\mathcal{I}_{v}$. \ Although $\mathcal{I}_{v}$ is
always somewhat longer than $\mathcal{I}_{c}$, in the football example it is
only slightly longer, thereby suggesting that the primary analysis is not
greatly distorted by the two versions of the control condition.%

\section*{References}%
%

\setlength{\hangindent}{12pt}
\noindent
Bachynski, K. E. (2016), \textquotedblleft Tolerable risks? physicians and
tackle football,\textquotedblright\ \textit{New England Journal of Medicine},
374, 405-407.%

\setlength{\hangindent}{12pt}
\noindent

\setlength{\hangindent}{12pt}
\noindent
Berger, R. L. (1982), \textquotedblleft Multiparameter hypothesis testing and
acceptance sampling,\textquotedblright\ \textit{Technometrics}, 24, 295-300.%

\setlength{\hangindent}{12pt}
\noindent
Campbell, D. T. (1969), \textquotedblleft Prospective: Artifact and
Control,\textquotedblright\ in R. Rosenthal and R. L. Rosnow, eds.,
\textit{Artifact in Behavioral Research}, New York: Academic Press, 351-382.%

\setlength{\hangindent}{12pt}
\noindent
Conover, W. J. and Salsburg, D. S. (1988), \textquotedblleft Locally most
powerful tests for detecting treatment effects when only a subset of patients
can be expected to respond to treatment,\textquotedblright%
\ \textit{Biometrics}, 44, 189-196.%

\setlength{\hangindent}{12pt}
\noindent
Deshpande, S.K., Hasegawa, R.B., Rabinowitz, A.R., Whyte, J., Roan, C.L.,
Tabatabaei, A., Baiocchi, M., Karlawish, J.H., Master, C.L., and Small, D.S.
(2017), \textquotedblleft High school football and later life cognition and
mental health: An observational study,\textquotedblright\ \textit{JAMA
Neurology}, 74, 909-918. \ Protocol for study 2016: \textsf{arXiv preprint
arXiv:1607.01756}.%

\setlength{\hangindent}{12pt}
\noindent
Fisher, R. A. (1935), \textit{The Design of Experiments}, Edinburgh: Oliver \& Boyd.%

\setlength{\hangindent}{12pt}
\noindent
Graves, A. B., White, E., Koepsell, T., Reier, B. V., Van Belle, G., Larson,
E. B., and Raskind, M. (1990), \textquotedblleft The association between head
trauma and Alzheimer's disease,\textquotedblright\ \textit{American Journal of
Epidemiology}, 131, 491-501.%

\setlength{\hangindent}{12pt}
\noindent
Hansen, B. B. (2004), \textquotedblleft Full matching in an observational
study of coaching for the SAT,\textquotedblright\ \ \textit{Journal of the
American Statistical Association}, 99, 609-618.%

\setlength{\hangindent}{12pt}
\noindent
Hansen, B. B. and Klopfer, S. O. (2006), \textquotedblleft Optimal full
matching and related designs via network flows,\textquotedblright%
\ \textit{Journal of Computational and Graphical Statistics}, 15, 609-627.
(\texttt{R} package \texttt{optmatch})%

\setlength{\hangindent}{12pt}
\noindent
Hansen, B. B. (2007), \textquotedblleft Flexible, optimal matching for
observational studies,\textquotedblright\ \textit{R News}, 7, 18-24.
(\texttt{R} package \texttt{optmatch})%

\setlength{\hangindent}{12pt}
\noindent
Hsu, J. Y. and Small, D. S. (2013), \textquotedblleft Calibrating sensitivity
analyses to observed covariates in observational studies,\textquotedblright%
\ \textit{Biometrics}, 69, 803-811.%

\setlength{\hangindent}{12pt}
\noindent
Knopman, D. S. and Ryberg, S. (1989), \textquotedblleft A verbal memory test
with high predictive accuracy for dementia of the Alzheimer
type,\textquotedblright\ \textit{Archives of Neurology}, 46, 141-145.%

\setlength{\hangindent}{12pt}
\noindent
Laska, E. M. and Meisner, M. J. (1989), \textquotedblleft Testing whether an
identified treatment is best,\textquotedblright\ \textit{Biometrics}, 45, 1139-1151.%

\setlength{\hangindent}{12pt}
\noindent
Lehman, E. J., Hein, M. J., Baron, S. L., and Gersic, C. M. (2012),
\textquotedblleft Neurodegenerative causes of death among retired national
football league players,\textquotedblright\ \textit{Neurology}, 79, 1970-1974.%

\setlength{\hangindent}{12pt}
\noindent
Lehmann, E. L. (1952), \textquotedblleft Testing multiparameter
hypotheses,\textquotedblright\ \textit{Annals of Mathematical Statistics}, 23, 541-552.%

\setlength{\hangindent}{12pt}
\noindent
Lehmann, E. L. and Romano, J. (2005), \textit{Testing Statistical Hypotheses}
(3$^{rd}$ edition), New York: Springer.%

\setlength{\hangindent}{12pt}
\noindent
Miles, S. H. and Prasad, S. (2016), \textquotedblleft Medical ethics and
school football,\textquotedblright\ \textit{American Journal of Bioethics},
16, 6-10.%

\setlength{\hangindent}{12pt}
\noindent
McKee, A. C., Cantu, R. C., Nowinski, C. J., Hedley-Whyte, T., Gavett, B. E.,
Budson, A. E., Santini, V. E., Lee, H.-S., Kublius, C. A., and Stern, R. A.
(2009), \textquotedblleft Chronic traumatic encephalopathy in athletes:
progressive tauopathy after repetitive head injury,\textquotedblright%
\ \textit{Journal of Neuropathology and Experimental Neurology}, 68, 709-735.%

\setlength{\hangindent}{12pt}
\noindent
Mortimer, J. A., van Duijn, C. M., Chandra, V., Fratiglioni, L., Graves, A.
B., Heyman, A., Jorm, A. F., Kokmen, E., Kondo, K., Rocca, W. A., Shalat, S.
L., Soininen, H., and for the Eurodem Risk Factors Research Group (1991),
\textquotedblleft Head trauma as a risk factor for alzheimer's disease: a
collaborative re-analysis of case-control studies,\textquotedblright%
\ \textit{International Journal of Epidemiology}, 20, S28-S35.%

\setlength{\hangindent}{12pt}
\noindent
Neyman, J. (1923, 1990), \textquotedblleft On the application of probability
theory to agricultural experiments,\textquotedblright\ {\itshape{Statistical
Science}}, 5, 463-480.%

\setlength{\hangindent}{12pt}
\noindent
Peto, R., Pike, M., Armitage, P., Breslow, N. E., Cox, D. R., Howard, S.V.,
Mantel, N., McPherson, K., Peto, J. and Smith, P.G. (1976), \textquotedblleft
Design and analysis of randomized clinical trials requiring prolonged
observation of each patient. I. Introduction and design,\textquotedblright%
\ \textit{British Journal of Cancer}, 34, 585-612.%

\setlength{\hangindent}{12pt}
\noindent
Pitman, E. J. (1937), \textquotedblleft Statistical tests applicable to
samples from any population,\textquotedblright\ \textit{Journal of the Royal
Statistical Society}, 4, 119-130.%

\setlength{\hangindent}{12pt}
\noindent
Rosenbaum, P. R. (1984), \textquotedblleft The consquences of adjustment for a
concomitant variable that has been affected by the
treatment,\textquotedblright\ \textit{Journal of the Royal Statistical
Society}, A, 147, 656-666.%

\setlength{\hangindent}{12pt}
\noindent
Rosenbaum, P. R. (1987), \textquotedblleft The role of a second control group
in an observational study,\textquotedblright\ \textit{Statistical Science}, 2, 292-306.%

\setlength{\hangindent}{12pt}
\noindent
Rosenbaum, P. R. (1991), \textquotedblleft A characterization of optimal
designs for observational studies,\textquotedblright\ \textit{Journal of the
Royal Statistical Society} B, 53, 597-610.%

\setlength{\hangindent}{12pt}
\noindent
Rosenbaum, P. R. (2007a), \textquotedblleft Confidence intervals for uncommon
but dramatic responses to treatment,\textquotedblright\ \textit{Biometrics},
63, 1164-1171.%

\setlength{\hangindent}{12pt}
\noindent
Rosenbaum, P. R\textsc{.} (2007b), \textquotedblleft Sensitivity analysis for
m-estimates, tests and confidence intervals in matched observational
studies,\textquotedblright\ \textit{Biometrics}, 63, 456-464. \ (\texttt{R}
package \texttt{sensitivitymult}; demonstration at
\textsf{https://rosenbap.shinyapps.io/learnsenShiny/})%

\setlength{\hangindent}{12pt}
\noindent
Rosenbaum, P. R. and Silber, J. H. (2009), \textquotedblleft Amplification of
sensitivity analysis in observational studies,\textquotedblright%
\ \textit{Journal American Statistical Association,} 104, 1398-1405.
\ (\texttt{amplify} function in the \texttt{R} package
\texttt{sensitivitymult})%

\setlength{\hangindent}{12pt}
\noindent
Rosenbaum, P. R. (2017), \textit{Observation and Experiment: An Introduction
to Causal Inference}, Cambridge, MA: Harvard University Press.%

\setlength{\hangindent}{12pt}
\noindent
Rosenthal, R., Rosnow, R. L., and Rubin, D. B. (2000), \textit{Contrasts
and effect sizes in behavioral research: A correlational approach},
New York, NY: Cambridge University Press.

\setlength{\hangindent}{12pt}
\noindent
Rubin, D. B. (1974), \textquotedblleft Estimating causal effects of treatments
in randomized and nonrandomized studies,\textquotedblright%
\ {\itshape{Journal of Educational
Psychology}}, 66, 688-701.%

\setlength{\hangindent}{12pt}
\noindent
Rubin, D. (1986), \textquotedblleft Comment: Which ifs have causal
answers,\textquotedblright\ \textit{Journal of the American Statistical
Association}, 81, 961-962.%

\setlength{\hangindent}{12pt}
\noindent
Rubin, D. B. (2007), \textquotedblleft The design versus the analysis of
observational studies for causal effects: parallels with the design of
randomized trials,\textquotedblright\ \textit{Statistics in Medicine}, 26, 20-36.%

\setlength{\hangindent}{12pt}
\noindent
Shaffer, J. P. (1974), \textquotedblleft Bidirectional unbiased
procedures,\textquotedblright\ \textit{Journal of the American Statistical
Association}, 69, 437-439.%

\setlength{\hangindent}{12pt}
\noindent
Steiger, J. H. (2004), \textquotedblleft Beyond the F Test: Effect Size Confidence
Intervals and Tests of CLose Fit in the Analysis of Variance and Contras Analysis, \textquotedblright\
\textit{Psychological Methods}, 9(2), 164-182.

\setlength{\hangindent}{12pt}
\noindent
Stuart, E. A. and Green, K. M. (2008), \textquotedblleft Using full matching
to estimate causal estimates in nonexperimental studies: examining the
relationship between adolescent marijuana use and adult
outcomes,\textquotedblright\ \textit{Developmental Psychology}, 44, 395-406.%

\setlength{\hangindent}{12pt}
\noindent
VanderWeele, T. J. and Hernan, M. A. (2013), \textquotedblleft Causal Inference
Under Multiple Versions of Treatment, \textquotedblright\ \textit{Journal of Causal Inference},
1(1), 1-20.

\setlength{\hangindent}{12pt}
\noindent
Welch, B. L. (1937), \textquotedblleft On the z-test in randomized blocks and
Latin squares,\textquotedblright\ \ \textit{Biometrika}, 29, 21-52.\newpage%


\section*{Appendix: Simulation comparing the power of the omnibus F-test to $\mathcal{I}_c$}
In this appendix, we compare the power of the omnibus F-test to the power of our version method to detect
departures from Fisher's sharp null hypothesis under varying degrees of ``versioning." The results can be
found in Table 2. When the degree of versioning is modest, the version method is more powerful than the F-test.
\newline

\begin{table}[ht!]
    \centering
    \begin{tabular}{l|rr|rr}
        \hline
        \multicolumn{1}{c|}{Version}& \multicolumn{2}{c}{$\tau^b = 0.25$} & \multicolumn{2}{|c}{$\tau^b = 0.4$}\\
        \hline
        & $P(\mathcal{I}_c \not\supseteq 0)$ & $P(F>c_{\alpha})$ & $P(\mathcal{I}_c \not\supseteq 0)$ & $P(F>c_{\alpha})$ \\
        \hline
        $\tau^a = \tau^b$ &0.61&0.49&0.94&0.91 \\
        $\tau^a = 0.95\times \tau^b$ &0.58&0.47&0.93&0.88 \\
        $\tau^a = 0.9\times \tau^b$ &0.56&0.47&0.93&0.89 \\
        $\tau^a = 0.75\times \tau^b$ &0.51&0.43&0.89&0.85 \\
        $\tau^a = 0.65\times \tau^b$ &0.44&0.42&0.83&0.84 \\
        $\tau^a = 0.6\times \tau^b$ &0.44&0.44&0.83&0.86 \\
        $\tau^a = 0.5\times \tau^b$ &0.37&0.41&0.75&0.84 \\
        $\tau^a = 0.25\times \tau^b$ &0.27&0.55&0.58&0.95 \\
        \hline
    \end{tabular}
    \caption{Comparison of the power of the version method (columns 2 and 4) and the power of the 
    F-test (columns 3 and 5) to reject Fisher's sharp null hypothesis under alternative hypotheses with
    varying degrees of ``versioning." Significance level of the tests are $\alpha=0.05$}
\end{table}

\noindent\textit{Simulation settings}

Let there be $I = 100$ matched sets, with $m_i=1$ treated and $n_i-m_i=4$ controls in each set. The design is 
balanced over versions, i.e., there are two controls of each version in each matched set. In each set we generate
outcomes as follows: $Y_{ij} = \tau + X_i + \epsilon_{ij}$ if the $j$-th subject in set $i$ receives treatment,
$Y_{ij} = X_i + \epsilon_{ij}$ if the $j$-th subject receives control version $b$, and $Y_{ij} = \delta + X_i + \epsilon_{ij}$ if the 
$j$-th subject receives control version. We let the individual and matched-set level terms, $\epsilon_{ij}$ and $X_i$,
be distributed as independent standard normals. Finally, let $\tau^b = \tau$ and $\tau^a = \tau-\delta$. If there are no versions, then
$\delta = 0$. When conducting the F-test of the null of no treatment effect, we model $X_i$ as a linear effect.
\end{document}